\documentclass[final]{siamltex}
\usepackage{mathrsfs}

\usepackage[linesnumbered,boxed,commentsnumbered,ruled]{algorithm2e}
\usepackage{graphicx}
\usepackage{bm}
\usepackage{latexsym}
\usepackage{amssymb}
\usepackage{epsfig}
\usepackage{indentfirst}
\usepackage{amsmath}
\usepackage{pstricks,pst-node,pst-tree}
\usepackage{float}
\usepackage[all]{xy}
\usepackage{url}

\newtheorem{claim}{Claim}
\newtheorem{fact}{Fact}

\title{Black-box Identity Testing for Low Degree Unmixed $\Sigma\Pi\Sigma\Pi(k)$ Circuits}

\author{Jinyu Huang\thanks{Department of Applied Mathematics, Illinois Institute of Technology,
         ({\tt jyhuangzju@gmail.com}).}}

\begin{document}

\maketitle

\begin{abstract}
A $\Sigma\Pi\Sigma\Pi(k)$ circuit
$C=\sum_{i=1}^kF_i=\sum_{i=1}^k\prod_{j=1}^{d_i}f_{ij}$ is unmixed
if for each $i\in[k]$, $F_i=f_{i1}(x_1)\cdots f_{in}(x_n)$, where
each $f_{ij}$ is a univariate polynomial given in the sparse
representation. In this paper, we give a polynomial time black-box
algorithm of identity testing for the low degree unmixed
$\Sigma\Pi\Sigma\Pi(k)$ circuits. In order to obtain the black-box
algorithm, we first show that a special class of low degree unmixed
$\Sigma\Pi\Sigma\Pi(k)$ circuits of size $s$ is $s^{O(k^2)}$-sparse.
Then we construct a hitting set $\mathcal{H}$ in polynomial time for
the low degree unmixed $\Sigma\Pi\Sigma\Pi(k)$ circuits from the
sparsity result above. The constructed hitting set is polynomial
size. Thus we can test whether the circuit or the polynomial $C$ is
identically zero by checking whether $C(a)=0$ for each
$a\in\mathcal{H}$. This is the first polynomial time black-box
algorithm for the low degree unmixed $\Sigma\Pi\Sigma\Pi(k)$
circuits, which also partly answers a question of Saxena \cite{SAX}.
\end{abstract}

\section{Introduction }

\label{sec1} A well known algebraic problem in algorithm design and
complexity theory is the Polynomial Identity Testing (PIT) problem:
given a multivariate polynomial $p(x_1,\cdots,x_n)$ over a field
$\mathbb{F}$, determine whether the polynomial is identically zero.
In many situations, an arithmetic circuit $C$ that computes the
polynomial $p(x_1,\cdots,x_n)$ is given as the input instead of the
polynomial $p(x_1,\cdots,x_n)$. Many other problems are related to
PIT. For example, primality testing \cite{AKS} or testing whether
there is a perfect matching \cite{LP} reduces to test whether a
particular polynomial is identically zero. Further, the proof of
$IP=PSPACE$ \cite{SHE} and the proof of the PCP theorem \cite{AS} in
complexity theory rely on the identity testing.

There is a randomized polynomial time algorithm for PIT, which was
given by Schwartz \cite{SCH} and Zippel \cite{ZIP}. Later, several
polynomial time randomized algorithms with fewer random bits were
introduced \cite{CK,LV}. But it is open to derandomize those
randomized polynomial time algorithms or design a deterministic
polynomial time or subexponential time algorithms for PIT. Kabanets
and Impagliazzo \cite{KI} proved that a polynomial time identity
testing algorithm implies that either $NEXP\not\subset P/poly$ or
Permanent is not computable by polynomial-size arithmetic circuits.
For the historic reason, it is hard to show the arithmetic circuit
lower bounds. Thus, researchers focus on PIT in some restricted
circuit models. Identity testing for sparse polynomials were studied
in \cite{KS}. There is a polynomial time algorithm if the sparsity
of the circuit is polynomial bounded. Deterministic algorithms for
some depth-$3$ circuits were known \cite{KS08,KS07}. Surveys
\cite{SAX,SY} have more information about the progress of the
identity testing.

Agrawal and Vinay \cite{AV} showed that a complete derandomization
of identity testing for depth-$4$ arithmetic circuits with
multiplication gates of small fanin implies a nearly complete
derandomization of general identity testing. As a result, it is
important and meaningful to study the depth-$4$ arithmetic circuits.
A polynomial $p(x_1,\cdots,x_n)$ of degree $poly(n)$ is called a low
degree polynomial. The arithmetic circuit $C$ that computes a low
degree polynomial is called a low degree circuit. Using the result
of Raz and Shpilka \cite{RS}, Saxena \cite{SAX08} gave a
deterministic white-box algorithm for depth-$4$ diagonal
$\Sigma\Pi\Sigma\Pi(k)$ circuits, which runs polynomial time for the
low degree circuits. Saraf and Volkovich \cite{SV11} recently
presented a deterministic black-box polynomial time algorithm for
the depth-$4$ multilinear $\Sigma\Pi\Sigma\Pi(k)$ circuits. Other
results concerned with depth-$4$ circuits can be found in \cite{SY}.

The known algorithm for $\Sigma\Pi\Sigma\Pi(k)$ circuits whose
multiplication gate have unmixed variables is non black-box. So it
is interesting whether there are black-box algorithms for them. In
fact, Sexena leaves this as an open problem in the survey
\cite{SAX}. In this paper, we resolve this problem for the low
degree $\Sigma\Pi\Sigma\Pi(k)$ circuits.

\subsection{Main Results}
\label{sec1.2}Similar to the result in \cite{SV11}, we first show
that each multiplication gate (in the second level) of the
pseudo-simple minimal low degree unmixed $\Sigma\Pi\Sigma\Pi(k)$
circuit is $s^{O(k^2)}$ sparse where $s$ is the size of the circuit.
Let $C=\sum_{i=1}^kF_i=\sum_{i=1}^k\prod_{j=1}^{d_i}f_{ij}$ be the
$\Sigma\Pi\Sigma\Pi(k)$ circuit. Roughly speaking, $C$ is
pseudo-simple if there is no $f_{ij}$ appears in all $F_i$. The
formal definition will be given later. The circuit $C$ is unmixed if
for each $i\in[k]$, $F_i=f_{i1}(x_1)\cdots f_{in}(x_n)$, where each
$f_{ij}$ is a univariate polynomial that is given in the sparse
representation. Based on this sparsity result, we obtain a
polynomial time black-box algorithm for the low degree unmixed
$\Sigma\Pi\Sigma\Pi(k)$ circuit.

\subsection{Outlines}
In section \ref{sec2}, we give required definitions, lemmas and
theorems. The sparsity bound for the low degree unmixed
$\Sigma\Pi\Sigma\Pi(k)$ circuit is given in section \ref{sec3}. We
present the black-box identity testing algorithm in section
\ref{sec4}.

\section{Preliminaries}

\label{sec2}

\subsection{Polynomials}
\label{subsec21}

The symbol $[n]$ denotes the set $\{1,\cdots,n\}$. Let $\mathbb{F}$
be the underlying field and let $\mathbb{\bar{F}}$ be its algebraic
closure. We assume that $\mathbb{F}$ contains sufficient number of
elements. Let $\mathbb{F}[x_1,\cdots,x_n]$ be a ring of polynomials
with coefficients in $\mathbb{F}$. Given a nonzero polynomial
$P\in\mathbb{F}[x_1,\cdots,x_n]$, it can be written in exactly one
way in the form

\begin{equation}\label{eq1}
P=\sum_{I}\alpha_{I}x^{I}
\end{equation}
where each coefficient $\alpha_{I}\neq0$ and
\begin{displaymath}
x^I:=x_1^{i_1}\cdots x_n^{i_n}
\end{displaymath}
with $I=(i_1,\cdots,i_n)$. The polynomial $P(x_1,\cdots,x_n)$
depends on variable $x_i$ if there are $c\in\mathbb{\bar{F}}^n$ and
$b\in\mathbb{\bar{F}}$ such that $P(c_1,\cdots,c_{i},\cdots,c_n)\neq
P(c_1,\cdots,b,\cdots,c_n)$. Then define $var(P):=\{i: \text{ P
depends on } x_i\}$. Let $P|_{x_A=a_A}$ be the polynomial with
$x_i=a_i$ for every $i\in A\subseteq[n]$. Given a multi-index
$I=(i_1,\cdots,i_n)$, define $I_{A=0}$ to be the multi-index by
setting $i_a=0$ for every $a\in A$. We define the sparsity of the
polynomial as follows.
\begin{definition}
The sparsity of the polynomial $P$ is the number of (nonzero)
monomials in $P$, which is represented by $||P||$. Given
$A\subseteq[n]$, define $||P||_A=|\{I_{A=0}: \alpha_I\neq0\}|$.
\end{definition}

If the polynomial contains a constant, assume that its multi-index
$I$ is $(0,\cdots,0)$. If $x_i^0$ is in a monomial, we can remove
$x_i^0$ from the monomial. There is an example for the sparsity. Let
$P=x_1^2x_2^5x_3+x_1^3x_2x_3^6+x_1-x_1$, then $||P||=2$ and
$||P||_A=2$ where $A=\{x_2,x_3\}$.

Given a subset $A=\{a_1,\cdots,a_k\}\subseteq[n]$ and a multi-index
$I=(i_1,\cdots,i_n)$, define $I_A:=(i_{a_1},\cdots,i_{a_k})$. We can
eliminate the variables with zero index from each monomial in the
polynomial. Then (\ref{eq1}) can also be written as
\begin{equation}\label{eq2}
P=\sum_{I_A}\alpha_{I_A}x^{I_A}
\end{equation}
where $x^{I_A}=x_{a_1}^{i_{a_1}}\cdots x_{a_k}^{i_{a_k}}$.

%\begin{lemma}
%Let $A\subseteq[n]$ and let $a\in\mathbb{\bar{F}}^n$ satisfy
%$(\prod_{f_{B,A}\not\equiv0}f_{B,A})|_{x_A=a_A}\not\equiv0$. Then
%$||f|_{x_A=a_A}||=||f||_A$ for the polynomial
%$f\in\mathbb{F}[x_1,\cdots,x_n]$.
%\end{lemma}

%\begin{proof}
%Since $||f|_{x_A=a_A}||\leq||f||_A$, it is sufficient to show that
%$||f|_{x_A=a_A}||\geq||f||_A$.

\begin{lemma}
Let $\{F_i\}$ and $\{G_i\}$ be sets of polynomials in
$\mathbb{F}[x_1,\cdots,x_n]$ with $F_i,G_i\not\equiv0$. Then
\begin{displaymath}
||gcd(F_1\cdot G_1,\cdots,F_k\cdot
G_k)||\leq||gcd(F_1,\cdots,F_k)||\cdot||G_1||\cdots||G_k||
\end{displaymath}
\end{lemma}

The proof is in \cite{SV11} (Observation $2.8$).

The following lemma is a corollary of the Shearer's Lemma.

\begin{lemma}\label{shear}
Let $P\in\mathbb{F}[x_1,\cdots,x_n]$ be a polynomial. Given $k$
disjoint sets $A_1,\cdots,A_k\subseteq[n]$ with $k\geq2$, we have
\begin{displaymath}
||P||^{k-1}\leq\prod_{j=1}^k||P||_{A_j}
\end{displaymath}
\end{lemma}
The proof is in \cite{SV11} (Corollary $2.6$). Now we define an
operator
\begin{definition}
Given $A\subseteq[n]$, $a\in\mathbb{\bar{F}}^n$ and
$P,Q\in\mathbb{F}[x_1,\cdots,x_n]$, define $D_{x_A=\alpha_A}(P,Q)$
as
\begin{displaymath}
D_{x_A=\alpha_A}(P,Q):=P\cdot
Q|_{x_A=\alpha_A}-P|_{x_A=\alpha_A}\cdot Q
\end{displaymath}
\end{definition}

\subsection{Low Degree Circuits}
\label{subsec22}

A $\Sigma\Pi\Sigma\Pi(k)$ circuit $C$ is a depth-$4$ circuit has
four alternating layers of addition and multiplication gates and the
number of input to the top addition gate is $k$. The
$\Sigma\Pi\Sigma\Pi(k)$ circuit $C$ with size $s$ computes a
polynomial in the form
\begin{displaymath}
C(x)=\sum_{i=1}^kF_i(x)=\sum_{i=1}^k\prod_{j=1}^{d_i}f_{ij}(x)
\end{displaymath}
where each $f_{ij}$ is a $s$-sparse polynomials (sparsity is at most
$s$). The circuit $C$ is called a low degree circuit, if it computes
a polynomial $P\in\mathbb{F}[x_1,\cdots,x_n]$ with degree at most
$poly(n)$. Without loss of generality, we can replace $poly(n)$ with
$O(n)$ or just $n$. For every $A\subseteq[k]$, define a subcircuit
of $C$ as $C_A:=\sum_{i\in A}F_i$. The circuit $C$ is minimal if
$C_A\not\equiv0$ for each $\emptyset\subsetneq A\subsetneq[k]$.
Given a polynomial $P\in\mathbb{F}[x_1,\cdots,x_n]$, the circuit $C$
is $P$-minimal if no proper subcircuit $C_A$ has an indecomposable
factor $P$.

Let $C$ be the $\Sigma\Pi\Sigma\Pi(k)$ circuit whose multiplication
gates have unmixed variables. Next, we define the pseudo greatest
common divisors for the unmixed polynomials. Suppose that
$F_1=f_{11}(x_1)\cdots f_{1n}(x_n)$ and $F_2=f_{21}(x_1)\cdots
f_{2n}(x_n)$ where each $f_{ij}(x_j)$ is a monic polynomial, let
$S_1:=\{f_{11}(x_1),\cdots,f_{1n}(x_n)\}$ and
$S_2:=\{f_{21}(x_1),\cdots,f_{2n}(x_n)\}$. Then the pseudo greatest
common divisor of $F_1$ and $F_2$ is
\begin{equation*}
\gcd(F_1,F_2)_{pseudo}=\prod_{f_i(x_i)\in S}f_i(x_i)
\end{equation*}
where $S=S_1\cap S_2$. Similarly, let
$S_i:=\{f_{i1}(x_1),\cdots,f_{in}(x_n)\}$ where
$F_i=f_{i1}(x_1)\cdot f_{i2}(x_2)\cdots\cdot f_{in}(x_n)$ for $1\leq
i\leq d$ ($d>2$). Define the pseudo greatest common divisor of
$F_1,\cdots,F_d$ for $d>2$ as
\begin{displaymath}
\gcd(F_1,\cdots,F_d)_{pseudo}=\prod_{f_i(x_i)\in S}f_i(x_i)
\end{displaymath}
where $S=\bigcap_{i=1}^dS_i$. Generally, $\gcd(F_1,F_2)_{pseudo}$ is
different from $\gcd(F_1,F_2)$. The unmixed circuit $C$ is
pseudo-simple if $\gcd(C)_{pseudo}=\gcd(F_1,\cdots,F_k)_{pseudo}=1$
where $C=\sum_{i=1}^kF_i$. The simplification of $C$ is defined as
$sim(C):=C/\gcd(C)_{pseudo}$. The polynomial
$P\in\mathbb{F}[x_1,\cdots,x_n]$ is said to be decomposable if it
can be written as $P(X)=Q(X_A)\cdot G(X_{\bar{A}})$ where
$\emptyset\subsetneq A\subsetneq[k]$. Otherwise, $P$ is called
indecomposable. The indecomposable factors of a polynomial
$P\in\mathbb{F}[x_1,\cdots,x_n]$ is $P=P_1(X_{I_1})\cdot
P_2(X_{I_2})\cdots,P_d(X_{I_d})$ such that all $\emptyset\subsetneq
I_j\subsetneq[n]$ are disjoint sets of indices and the $P_i$-s are
indecomposable for $1\leq i\leq d$. Let $f|_{in}C$ denote that f is
an indecomposable factor of $C$. The following lemma will be used to
prove the sparsity bound.

\begin{lemma}
\label{lem26} Let $P$ be a non-constant univariate low-degree
polynomial and let $Q\in\mathbb{F}[x_1,\cdots,x_n]$ be a low-degree
polynomial. Further, let $c\in\mathbb{\bar{F}}$ where $P(c)\neq0$.
Suppose that $var(P)=\{i\}$, then $D_{x_i=c}(P,Q)\equiv0$ if and
only if $P|_{in}Q$.
\end{lemma}

\begin{proof}
Suppose that $D_{x_i=c}(P,Q)\equiv0$. Since $P(c)\neq0$, we have
$P(c)=a$ where $a\in\mathbb{\bar{F}}$ is a nonzero element of the
field. Thus from $D_{x_i=c}(P,Q)\equiv0$, we have
\begin{displaymath}
P\cdot Q|_{x_i=c}= a\cdot Q
\end{displaymath}
Since $Q|_{x_i=c}$ does not depend on $x_i$ and $P$ is a univariate
polynomial, $P|_{in}Q$.

Now suppose that $P|_{in}Q$, we show $D_{x_i=c}(P,Q)\equiv0$. Since
$P|_{in}Q$, we have $Q=P\cdot H$ where $P$ and $H$ are variable
disjoint factors of $Q$. Further, $P(c)=a$ for some nonzero
$a\in\mathbb{\bar{F}}$. Thus
\begin{eqnarray*}
D_{x_i=c}(P,Q) &=& P\cdot Q|_{x_i=c}-P|_{x_i=c}\cdot Q\\
               &=& P\cdot a\cdot H-a\cdot P\cdot H\\
               &\equiv& 0
\end{eqnarray*}
\end{proof}

The following lemma characterize the pseudo greatest common
divisors.

\begin{lemma}\label{lem27}
Let $P_1,\cdots,P_k\in\mathbb{F}[x_1,\cdots,x_n]$ be non-constant
low degree unmixed polynomials. Let $P_i=\beta_i\cdot
g_{i1}(x_1)\cdot g_{i2}(x_2)\cdots g_{in}(x_n)$ such that
$\beta_i\in\mathbb{F}$ and $g_{i1}$, $g_{i2},\cdots,g_{in}$ are
univariate monic polynomials for $2\leq i\leq k$. Let
$P_1=\beta\cdot g_1(x_1)\cdot g_2(x_2)\cdots g_n(x_n)$ such that
$\beta$ is a field element and $g_1$, $g_2,\cdots,g_n$ are
univariate monic polynomials. Then
$\gcd(P_1,\cdots,P_k)_{pseudo}\neq1$ if and only if there exists a
nonconstant factor $g_d(x_d)$($1\leq d\leq n$) of $P_1$ such that
$\prod_{1\leq j\leq n}(g_d-g_{ij})\equiv0$ for each $2\leq i\leq k$.
\end{lemma}

\begin{proof}
Suppose $G:=\gcd(P_1,\cdots,P_k)_{pseudo}\neq1$, then there exists a
factor $g_d(x_d)$ of $P_1$ such that $g_d(x_d)|_{in}G$. Hence,
$g_d(x_d)$ is a common factor of $P_1, P_2,\cdots,P_k$. Then there
exists $g_{id}(x_d)$ such that $g_d(x_d)-g_{id}(x_d)\equiv0$ for
each $2\leq i\leq k$. Thus we have $\prod_{1\leq j\leq
n}(g_d-g_{ij})\equiv0$ for each $2\leq i\leq k$.

Now suppose that there exists a factor $g_d(x_d)$($1\leq d\leq n$)
of $P_1$ such that $\prod_{1\leq j\leq n}(g_d-g_{ij})\equiv0$ for
each $2\leq i\leq k$. Then there exists a $j$ such that $g_d\equiv
g_{ij}$ for each $2\leq i\leq k$. Thus $g_d$ is a common factor of
$P_1, P_2,\cdots,P_k$. As a result, $g_d$ is a factor of
$\gcd(P_1,\cdots,P_k)_{pseudo}$. Since $g_d\neq1$, we have
$\gcd(P_1,\cdots,P_k)_{pseudo}\neq1$.
\end{proof}

\subsection{Hitting Sets and Generators}
\label{subsec23}

A set $\mathcal{H}\subseteq\mathbb{F}^n$ is a hitting set for a
circuit class $\mathcal{M}$, if given any non-zero circuit
$P\in\mathcal{M}$, there exists $a\in\mathcal{H}$ such that
$P(a)\neq0$. A generator for the circuit class $\mathcal{M}$ is a
polynomial mapping
$\mathcal{G}=(\mathcal{G}_1,\cdots,\mathcal{G}_n):\mathbb{F}^m\rightarrow\mathbb{F}^n$
such that for each nonzero $n$-variate polynomial $P\in\mathcal{M}$,
we have $P(\mathcal{G})\not\equiv0$. In the identity testing,
generators and hitting sets play the same role. The following lemma
is about the generator for the low degree polynomials.

\begin{lemma}
\label{lem1} There is a generator
$\mathcal{L}_s:=(\mathcal{L}_{1,s},\cdots,\mathcal{L}_{n,s}):\mathbb{F}^q\rightarrow\mathbb{F}^n$
for $s$-sparse low degree polynomials. The individual degrees of
every $\mathcal{L}_{i,s}$ are bounded by $n-1$ and $q=O(\log_ns)$.
\end{lemma}

The Lemma \ref{lem1} follows from the following two facts.

\begin{fact}
\label{fact1} We have a hitting set $\mathcal{H}$ with cardinality
$poly(n,s,d)$ for each non-zero $n$-variate $s$-sparse polynomial
with degree $d$ over a field $\mathbb{F}$.
\end{fact}

The statement can be found in \cite{KS}. In particular, for each
non-zero low degree $n$-variate $s$-sparse polynomial, there exists
a hitting set with cardinality $poly(n,s)$.

\begin{fact}
\label{fact2} Let $|\mathbb{F}|>n$. Given a hitting set
$\mathcal{H}\subseteq\mathbb{F}^n$ for a circuit class
$\mathcal{M}$, there is a $poly(|\mathcal{H}|,n)$ time algorithm
that produce a generator
$\mathcal{G}:\mathbb{F}^q\rightarrow\mathbb{F}^n$ for $\mathcal{M}$.
The individual degrees of each $\mathcal{G}_i$ is bounded by $n-1$
and $q=\lceil\log_n|\mathcal{H}|\rceil$.
\end{fact}

The proof can be found in \cite{KMSV}. Combining the Fact
\ref{fact1} and the Fact \ref{fact2}, we have the Lemma \ref{lem1}.
Some additional facts are needed in the paper.

\begin{fact}\label{fact3}
Let $P=P_1\cdot P_2\cdots P_n$ be a product of nonzero polynomials
where $P_i\in\mathcal{M}$ for each $i$. Let $\mathcal{G}$ be a
generator for $\mathcal{M}$, then $P(\mathcal{G})\not\equiv0$.
\end{fact}

Another fact is given in \cite{ALON}.

\begin{fact}\label{fact4}
Let $P=P(x_1,\cdots,x_n)$ be a polynomial in $n$ variables over an
arbitrary field $\mathbb{F}$. Suppose that the degree of $P$ in
$x_i$ is bounded by $d_i$ for $1\leq i\leq n$. Let
$S_i\subseteq\mathbb{F}$ be a set with at leat $d_i+1$ elements of
$\mathbb{F}$. If $P=0$ for all $n$-tuples
$(a_1,\cdots,a_n)\in\prod_{i=1}^nS_i$, then $P\equiv0$.
\end{fact}

\section{Upper Bound of the Sparsity}

\label{sec3}

In this section, we give a sparsity upper bound for the
pseudo-simple minimal low degree unmixed $\Sigma\Pi\Sigma\Pi(k)$
circuits that computes the zero polynomial. The proof is similar to
that in \cite{SV11} with some nontrivial modifications.

\begin{theorem}
\label{mainbound} Let $C$ be a pseudo-simple minimal low degree
unmixed $\Sigma\Pi\Sigma\Pi(k)$ circuit of size $s$, which can be
written as $C(x)=\sum_{i=1}^kF_i(x)$. If $C$ computes a zero
polynomial, then $||F_i||\leq s^{5k^2}$ for each $i\in[k]$.
\end{theorem}

\begin{proof}
We prove the statement by induction on $k$. Let $k=2$ be the basis
case. Since $C$ is pseudo-simple and minimal, We have $C=c-c$ for
some unit $c\in\mathbb{F}$. Thus the sparsity of $F_1$ and $F_2$ is
one. Suppose that the statement is true for $2\leq k\leq K-1$, we
show that the statement is true for $k=K$. The proof is based on
three claims.

\begin{claim}
\label{cl311} Let $G:=\gcd(F_1,\cdots,F_t)_{pseudo}$ be the pseudo
greatest common divisor of $\{F_1,\cdots,F_t\}$, where $2\leq t\leq
k-1$. The sparsity of $G$ satisfies $||G||\leq s^{5(k-t+1)^2}$.
\end{claim}
\begin{proof}
Let $V:=[n]-var(G)$, then we can write $F_i$ for $1\leq i\leq t$ as
$F_i=G\cdot g_i$ where $var(g_i)\cap var(G)=\emptyset$. Further, we
have $F_i=f_{i1}(x_1)\cdot f_{i2}(x_2)\cdots f_{in}(x_n)$ for $1\leq
i\leq k$. Let $M_{f_{1d},F_i}:=\prod_{1\leq j\leq n}(f_{1d}-f_{ij})$
for all $1\leq d\leq n$ and $2\leq i\leq k$. Now define the
polynomial
\begin{displaymath}
\Phi:=\prod_{\emptyset\subsetneq
A\subsetneq[k]}C_A\prod_{1d,i:M_{f_{1d},F_i}\not\equiv0}M_{f_{1d},F_i}
\end{displaymath}
Since $\mathbb{F}$ is sufficient large, there exists an element
$a\in\mathbb{\bar{F}}^n$ such that $\Phi(a)\neq0$. Hence
$\Phi|_{x_V=a_V}\not\equiv0$. Let $F^{\prime}_i=F_i|_{x_V=a_V}$ for
each $i$ and let $C^{\prime}=C|_{x_V=a_V}=\sum_{i=1}^kF^{\prime}_i$.
We show that $C^{\prime}$ is pseudo-simple and minimal. Since
$C_A|_{x_V=a_V}\not\equiv0$ for each nonempty proper subset $A$ of
$[k]$, $C^{\prime}$ is minimal. Since $C$ is pseudo-simple, we have
$\gcd(F_1,\cdots,F_k)_{pseudo}=1$. Hence for every non-constant
factor $f_{1d}$ of $F_1$ there exists $2\leq i\leq k$ such that
$M_{f_{1d},F_i}\not\equiv0$ by the Lemma \ref{lem27}. Let
$f^{\prime}_{1d}=f_{1d}|_{x_V=a_V}$ for $1\leq d\leq n$. Then for
each $f_{1d}$ of $F_1$, there exists $2\leq i\leq k$ such that
$M_{f^{\prime}_{1d},F^{\prime}_i}\not\equiv0$. Thus
$\gcd(F^{\prime}_1,\cdots,F^{\prime}_k)_{pseudo}=1$. Define
$P_1:=\sum_{i=1}^tF^{\prime}_i$, then
\begin{displaymath}
P_1=(\sum_{i=1}^tg_i|_{x_V=a_V})\cdot G=c\cdot G
\end{displaymath}
where $c\in\mathbb{\bar{F}}$. Let $P_i:=F^{\prime}_{t+i-1}$ for
$2\leq i\leq k-t+1$. Now we define a new low degree unmixed
$\Sigma\Pi\Sigma\Pi(k)$ circuit
$C^{\prime\prime}:=\sum_{i=1}^{k-t+1}P_i$. At first,
$C^{\prime\prime}\equiv0$, since $C^{\prime}\equiv0$. Because
$C^{\prime}$ is minimal, $c\cdot G=C^{\prime}_{[t]}\not\equiv0$ and
$C^{\prime\prime}_A\not\equiv0$ for any nonempty subset $1\not\in
A\subseteq[k-t+1]$. So $C^{\prime\prime}$ is minimal. Finally,
$C^{\prime\prime}$ is pseudo-simple, since
\begin{displaymath}
\gcd(C^{\prime\prime})_{pseudo}=\gcd(F^{\prime}_1,\cdots,F^{\prime}_k)_{pseudo}=1
\end{displaymath}
Now we can apply the induction hypothesis, since $k-t+1<k$. As a
result, $||G||=||P_1||\leq s^{5(k-t+1)^2}$.
\end{proof}

Let $f$ be be a univariate polynomial, the arithmetic circuit is
called $f$-minimal if no proper subcircuit has an indecomposable
factor $f$. Moreover, recall that $f|_{in}C$ means that $f$ is an
indecomposable factor of $C$.

\begin{claim}\label{cl312}
Let $C=\sum_{i=1}^kF_i$ be a low degree unmixed
$\Sigma\Pi\Sigma\Pi(k)$ circuit computing the zero polynomial and
let $f$ be a univariate polynomial such that $f|_{in}F_k$ and
$f\nmid_{in} F_1$. There exists a subset $B\subseteq[k]$ with
$2\leq|B|\leq k-1$ satisfying: the subcircuit $C_B$ is $f$-minimal,
$f|_{in}C_B$ and $1\in B$.
\end{claim}
\begin{proof}
Since $C$ computes the zero polynomial, we have
\begin{displaymath}
F_k=-\sum_{i=1}^{k-1}F_i
\end{displaymath}
Further, since $f|_{in}F_k$, $f$ is an indecomposable factor of
$\sum_{i=1}^{k-1}F_i$. If $\sum_{i=1}^{k-1}F_i$ is not minimal, we
can continuously partition $\sum_{i=1}^{k-1}F_i$ into $f$-minimal
circuits such that $f|_{in}C_A$ for each $f$-minimal subcircuit
$C_A$ in the partition and $F_i$ is contained in only one subcircuit
of the partition for each $i\in[k-1]$. Suppose that $C_B$ is the
$f$-minimal subcircuit such that $1\in B$ and $f|_{in}C_B$. Since
$k\not\in B$ and $f\nmid_{in} F_1$, we have $2\leq|B|\leq k-1$.
\end{proof}

\begin{claim}\label{cl313}
Let $C=\sum_{i=1}^tF_i$ be a pseudo-simple $f$-minimal low-degree
unmixed $\Sigma\Pi\Sigma\Pi(k)$ circuit of size $s$ where $2\leq
t\leq k-1$. Let $f$ be a non-constant univariate low-degree
polynomial. Suppose that $f|_{in}C$, then $||F_1||_{var(f)}\leq
s^{5t^2}$.
\end{claim}
\begin{proof}
Suppose that $var(f)=\{d\}$ where $1\leq d\leq n$ and there exists
an element $a\in\mathbb{\bar{F}}$ such that $f(a)\neq0$. Let
$C^{\prime}:=D_{x_d=a}(f,C)=\sum_{i=1}^tD_{x_d=a}(f,F_i)$. Since
$F_i$ is an unmixed polynomial, we have $F_i=P_i\cdot f_i$ where
$d\not\in var(P_i)$. If $f_{id}(x_d)\neq1$ for $F_i$, then set
$f_i:=f_{id}(x_d)$. Otherwise let $f_i:=1$. Then by the Lemma
\ref{lem26}, we have
\begin{displaymath}
C^{\prime}=\sum_{i=1}^tP_i\cdot D_{x_d=a}(f,f_i)\equiv0
\end{displaymath}
Since $C^{\prime}$ is $f$-minimal,
$C^{\prime}_A=D_{x_d=a}(f,C^{\prime}_A)\not\equiv0$ for each
nonempty subset $A\subseteq[t]$. Thus $C^{\prime}$ is minimal.
Further, since $D_{x_d=a}(f,f_i)$ is a constant, $C^{\prime}$ is a
pseudo-simple low-degree unmixed $\Sigma\Pi\Sigma\Pi(k)$ circuit.
Because $t\leq k-1$, we can apply the induction hypothesis. So
$||F_1||_{var(f)}=||P_1||\leq s^{5t^2}$.
\end{proof}

Now we can prove the induction step by contradiction. Without loss
of generality, assume that $||F_k||>s^{5k^2}$. We show that
$||F_i||\leq s^{5k^2-1}$ for $1\leq i\leq k-1$. Since the circuit is
symmetric, it is sufficient to prove that $||F_1||\leq s^{5k^2-1}$.
Let $P_k:=F_k/\gcd(F_1,F_k)_{pseudo}$. Without loss of generality,
assume that $P_k=f_{k1}(x_1)\cdot f_{k2}(x_2)\cdots f_{km}(x_m)$.
Then we have
\begin{displaymath}
||P_k||=||f_{k1}(x_1)\cdot f_{k2}(x_2)\cdots f_{km}(x_m)||\leq s^m
\end{displaymath}
Further, we have
\begin{displaymath}
||P_k||\geq\frac{||F_k||}{||\gcd(F_1,F_k)_{pseudo}||}>s^{5k^2-5(k-1)^2}
\end{displaymath}
by the Claim \ref{cl311}. Thus $m\geq 10k-5$. Given any
$f_{kj}(x_j)$ of $P_k$ for $1\leq j\leq m$, there exists a subset
$A\subseteq[k]$ such that $1\in A$, $C_A$ is $f_{kj}$-minimal and
$f_{kj}|_{in}C_A$ by the Claim \ref{cl312}. Assume that
$A=\{1,2,\cdots,t\}$ where $2\leq t\leq k-1$. Let
$G=\gcd(F_1,\cdots,F_t)_{pseudo}$, then $||G||\leq s^{5(k-t+1)^2}$
by the Claim \ref{cl311}. Since $C_A$ is $f_{kj}$-minimal,
$f_{kj}\nmid_{in}F_i$ for every $i\in[t]$. Thus $f_{kj}\nmid_{in}G$.
Then $C^{\prime}_A:=\sum_{i=1}^tF_i/G$ is a pseudo-simple
$f_{kj}$-minimal low-degree unmixed $\Sigma\Pi\Sigma\Pi(t)$ circuit.
Since $f_{kj}|_{in}C^{\prime}_A$, we have
$||F_1/G||_{var(f_{kj})}\leq s^{5t^2}$ by the Claim \ref{cl313}. As
a result, we have
\begin{displaymath}
||F_1||_{var(f_{kj})}\leq||F_1/G||_{var(f_{kj})}\cdot||G||\leq
s^{5(k-1)^2+k+19}
\end{displaymath}
The above inequality is valid for each $1\leq j\leq 10k-5$.
Moreover, all $f_{kj}$-s share no variable. Then applying the Lemma
\ref{shear}, we have
\begin{displaymath}
||F_1||\leq\Big(\prod_{j=1}^{10k-5}||F_1||_{var(f_{kj})}\Big)^{\frac{1}{10k-6}}<s^{5k^2-1}
\end{displaymath}
Since the circuit is symmetric, $||F_i||<s^{5k^2-1}$ holds for each
$i\in[k-1]$. Hence,
\begin{displaymath}
||F_k||=||\sum_{i=1}^{k-1}F_i||\leq\sum_{i=1}^{k-1}||F_i||<s^{5k^2}
\end{displaymath}
This leads to a contradiction. So $||F_i||\leq s^{5k^2}$ for each
$i\in[k]$.
\end{proof}

\section{The Black-Box Algorithm}

\label{sec4} Similar to \cite{SV11}, we construct a generator for
low degree unmixed $\Sigma\Pi\Sigma\Pi(k)$ circuits. The image of
the generator is the hitting set for such circuits. Then a
polynomial time black-box algorithm can be obtained from the
generator. Fix a set $C=\{c_0,c_1,\cdots,c_n\}\subseteq\mathbb{F}$
with $n+1$ distinct elements. Recall that $\mathcal{L}_m$ is a
generator for $m$-sparse low-degree polynomials. Let
$\overrightarrow{y_i}$ denote the vector with $q$ entries for each
$i$.
\begin{definition}
For each $i\in[n]$ let $W_i(z):\mathbb{F}\rightarrow\mathbb{F}$ be
the degree $n$ polynomial such that $W_i(c_j)=1$ for $j\geq i$ and
$W_i(c_j)=0$ otherwise. For each $l\geq1$, $m\geq1$ and $i\in[n]$,
define
\begin{displaymath}
\mathcal{S}^i_{l,m}(\overrightarrow{y_1},\cdots,\overrightarrow{y_l},z_1,\cdots,z_l):=\mathcal{S}^i_{l-1,m}(\overrightarrow{y_1}
,\cdots,\overrightarrow{y_{l-1}},z_1,\cdots,z_{l-1})\cdot
W_i(z_l)+\mathcal{L}^i_m(\overrightarrow{y_l})(1-W_i(z_l))
\end{displaymath}
where
$\mathcal{S}^i_{l,m}(\overrightarrow{y_1},\cdots,\overrightarrow{y_l},z_1,\cdots,z_l)$
is a function from $\mathbb{F}^{q\cdot l+l}$ to $\mathbb{F}$.
Moreover, let $\mathcal{S}^i_{0,m}\equiv0$ for any $i$. Then define
$\mathcal{S}_{l,m}:\mathbb{F}^{q\cdot l+l}\rightarrow\mathbb{F}^n$
as
\begin{displaymath}
\mathcal{S}_{l,m}:=(\mathcal{S}^1_{l,m},\mathcal{S}^2_{l,m},\cdots,\mathcal{S}^n_{l,m})
\end{displaymath}
\end{definition}

Now we have the following fact from the definition.

\begin{fact}\label{fact5}
For each $0\leq d\leq n$, we have
\begin{displaymath}
\mathcal{S}_{l,m}\mid_{z_l=c_d}=(\mathcal{S}^1_{l-1,m},\cdots,\mathcal{S}^d_{l-1,m},\mathcal{L}^{d+1}_m,\cdots,\mathcal{L}^n_m)
\end{displaymath}
Then for every $\vec{a}\in Im(\mathcal{S}_{l-1,m})$ and $\vec{b}\in
Im(\mathcal{L}_m)$, we have $(a_1,\cdots,a_d,b_{d+1},\cdots,b_n)\in
Im(\mathcal{S}_{l,m})$. In particular, $Im(\mathcal{S}_{l-1,m})\cup
Im(\mathcal{L}_m)\subseteq Im(\mathcal{S}_{l,m})$.
\end{fact}

Next, we show that there is a generator for the low degree unmixed
$\Sigma\Pi\Sigma\Pi(k)$ circuits

\begin{lemma}\label{gen}
Suppose that $P\in\mathbb{F}[x_1,\cdots,x_n]$ is a nonzero
polynomial computed by a low degree unmixed $\Sigma\Pi\Sigma\Pi(k)$
circuit $C=\sum_{i=1}^kF_i=\sum_{i=1}^k\prod_{j=1}^{d_i}f_{ij}(x_j)$
with size $s\geq2$ and $k\geq1$. Then it holds that
$P(\mathcal{S}_{k,m})\not\equiv0$ for every $m\geq s^{5k^2+2}$.
\end{lemma}

\begin{proof}
At first, we need some claims that are needed in the proof.

\begin{claim}
\label{cl411} Let $M\geq1$ and let $C=\sum_{i=1}^kF_i$ be a low
degree unmixed $\Sigma\Pi\Sigma\Pi(k)$ circuit of size $s$ such that
$\max_i||F_i||>M$. Let $a\in\mathbb{F}^n$ with $F_i(a)\neq0$ for
each $i\in[k]$. Then there is a $0\leq t\leq n-1$ such that
$M<\max_i||F_i\mid_{x_{[t]}=a_{[t]}}||\leq M\cdot s$.
\end{claim}

\begin{proof}
Since $\max_i||F_i||>M$ and
$\max_i||F_i\mid_{x_{[n]}=a_{[n]}}||\leq1$, let $t$ be the largest
index such that $M<\max_i||F_i\mid_{x_{[t]}=a_{[t]}}||$. Thus we
have $\max_i||F_i\mid_{x_{[t+1]}=a_{[t+1]}}||\leq M$. Since $F_i$ is
an unmixed polynomial and $F_i(a)\neq0$, setting $x_{t+1}$ to
$a_{t+1}$ can affect at most one factor $f_{ij}(x_j)$ (where
$j=t+1$) for each $F_i$. As a result, it can reduce the sparsity by
a factor at most $||f_{ij}(x_j)||\leq s$. Hence
$||F_i\mid_{x_{[t]}=a_{[t]}}||\leq M\cdot s$.
\end{proof}

\begin{claim}\label{cl412}
Let $P\in\mathbb{F}[x_1,\cdots,x_n]$ be a nonzero polynomial
computed by a pseudo-simple minimal low-degree unmixed
$\Sigma\Pi\Sigma\Pi(k)$ circuit $C=\sum_{i=1}^kF_i$ with size $s$
and $k\geq2$. In addition, let $\mathcal{G}_{k-1}$ be a generator
for $\Sigma\Pi\Sigma\Pi(k-1)$ circuits of size $s$ and $s$-sparse
polynomials. Then there is a $c\in Im(\mathcal{G}_{k-1})$ and $0\leq
t\leq n-1$ such that $P^{\prime}:=P\mid_{x_{[t]}=c_{[t]}}$ is a
nonzero $s^{5k^2+2}$-sparse polynomials.
\end{claim}

\begin{proof}
If $\max_i||F_i||\leq s^{5k^2}$, then $||P||\leq k\cdot s^{5k^2}\leq
s^{5k^2+2}$. Otherwise, we have $\max_i||F_i||> s^{5k^2}$. Recall
that $F_i=f_{i1}(x_1)\cdot f_{i2}(x_2)\cdots f_{in}(x_n)$ for $1\leq
i\leq k$. Let $M_{f_{1d},F_i}:=\prod_{1\leq j\leq n}(f_{1d}-f_{ij})$
for all $1\leq d\leq n$ and $2\leq i\leq k$. Now define the
polynomial
\begin{displaymath}
\Phi:=\prod_{\emptyset\subsetneq
A\subsetneq[k]}C_A\prod_{1d,i:M_{f_{1d},F_i}\not\equiv0}M_{f_{1d},F_i}
\end{displaymath}
Since each multiplicand of $\Phi$ is either a $s$-sparse polynomial
or a $\Sigma\Pi\Sigma\Pi(k-1)$ circuit,
$\Phi(\mathcal{G}_{k-1})\not\equiv0$. Thus there is a $c\in
Im(\mathcal{G}_{k-1})$ such that $\Phi(c)\neq0$. Since $F_i$ appears
in the multiplicands of $\Phi$, $F_i(c)\neq0$ for each $i\in[k]$.
Then we have
\begin{displaymath}
s^{5k^2}<\max_i||F_i\mid_{x_{[t]}=c_{[t]}}||\leq s^{5k^2+1}
\end{displaymath}
by the Claim \ref{cl411}. Let
$C^{\prime}:=C\mid_{x[t]=c[t]}=\sum_{i=1}^kF_i\mid_{x_{[t]}=c_{[t]}}$.
Similar to the argument in the proof of the Theorem \ref{mainbound},
$C^{\prime}$ is pseudo-simple and minimal. Further, we have
$C^{\prime}\not\equiv0$ by the Theorem \ref{mainbound} and
$\max_i||F_i||> s^{5k^2}$. Finally,
$||C^{\prime}||=||P\mid_{x_{[t]}=c_{[t]}}||\leq s^{5k^2}\cdot s\cdot
k\leq s^{5k^2+2}$.
\end{proof}

Now we prove the statement by induction on $k$. If $k=1$, then $P$
is a product of $s$-sparse polynomials. Thus by the Fact \ref{fact3}
and the the Fact \ref{fact5}, $P(\mathcal{S}_{1,m})\not\equiv0$.
Assume that the statement is true for $2\leq k\leq K$. We can assume
that $C$ is pseudo-simple and minimal. If $C$ is not minimal, there
is a $\Sigma\Pi\Sigma\Pi(k-1)$ circuit $C^{\prime}$ computing $P$.
Then the induction hypothesis can be applied. If $C$ is not
pseudo-simple, we have $P=G\cdot C^{\prime}$ where
$G=\gcd(C)_{pseudo}$ and $C^{\prime}=Sim(C)$. Since $G$ is a product
of $s$-sparse polynomials, $G(\mathcal{S}_{k,m})\not\equiv0$ by the
reason that is identical to the base case. So without loss of
generality, we can assume that $C$ is pseudo-simple and minimal. By
the induction hypothesis, $\mathcal{S}_{k-1,m}$ is a generator for
$\Sigma\Pi\Sigma\Pi(k-1)$ circuits and it is also a generator for
$s$-sparse polynomials. Then from the Claim \ref{cl412}, there is a
$c\in Im(\mathcal{S}_{k-1,m})$ and $0\leq t\leq n-1$ such that
$P^{\prime}=P\mid_{x_{[t]}=c_{[t]}}$ is a nonzero
$s^{5k^2+2}$-sparse polynomial. Since $\mathcal{L}_m$ is a generator
for $s^{5k^2+2}$-sparse polynomials, there exists a $b\in
Im(\mathcal{L}_m)\subseteq Im(\mathcal{S}_{k,m})$ such that
$P^{\prime}(b)\neq0$. As a result,
$P(a_1,\cdots,a_t,b_{t+1},\cdots,b_n)\neq0$.
\end{proof}

Then a black-box algorithm can be obtained by constructing a hitting
set for the low degree unmixed $\Sigma\Pi\Sigma\Pi(k)$ circuit.

\begin{algorithm}\label{alg1}
\SetKwData{Left}{left}\SetKwData{This}{this}\SetKwData{Up}{up}
\SetKwFunction{Union}{Union}\SetKwFunction{FindCompress}{FindCompress}
\SetKwInOut{Input}{input}\SetKwInOut{Output}{output}
\Input{A low degree unmixed $\Sigma\Pi\Sigma\Pi(k)$ circuit $C$}
\Output{A hitting set $\mathcal{H}$}\Begin{ Construct a set
$U\subseteq\mathbb{F}$ with size $n^3+1$\;
$\mathcal{H}:=\mathcal{S}_{k,s^{5k^2+2}}((U^q)^k\times U^k)$\;
Return $\mathcal{H}$\;}\caption{Construct a Hitting Set}
\end{algorithm}

\begin{theorem}\label{algo41}
The Algorithm \ref{alg1} is a polynomial time algorithm, which
outputs a hitting set $\mathcal{H}$ of size $n^{O(k)}\cdot
s^{O(k^3)}$ for the low degree unmixed $\Sigma\Pi\Sigma\Pi(k)$
circuit with size $s$.
\end{theorem}

\begin{proof}
Suppose that $P\in\mathbb{F}[x_1,\cdots,x_n]$ is a nonzero low
degree polynomial computed by the low degree unmixed
$\Sigma\Pi\Sigma\Pi(k)$ circuit of size $s$. From the Lemma
\ref{gen}, $P(\mathcal{S}_{k,s^{5k^2+2}})\not\equiv0$ that depends
on $(q+1)\cdot k=k+k(5k^2+2)\log_ns$ variables. Since the individual
degree is less than $n^3+1$ for $P(\mathcal{S}_{k,s^{5k^2+2}})$,
there exists a $c\in\mathcal{H}$ such that $P(c)\neq0$ by the fact
\ref{fact4}. So $\mathcal{H}$ is a hitting set of $P$. The size of
$\mathcal{H}$ is
\begin{displaymath}
|\mathcal{H}|\leq n^{O(k)}\cdot n^{O(k^3\log_ns)}=n^{O(k)}\cdot
s^{O(k^3)}
\end{displaymath}
Since $k$ is a constant and the generator for sparse polynomials can
be constructed in polynomial time by the Fact \ref{fact2}, the
generator $\mathcal{S}_{k,s^{5k^2+2}}$ can be constructed in
polynomial time. Then it is obvious that the Algorithm \ref{alg1} is
a polynomial time algorithm.
\end{proof}

\section{Conclusions and an Open Problem}
We give a polynomial time black-box algorithm of identity testing
for the low degree unmixed $\Sigma\Pi\Sigma\Pi(k)$ circuits. An open
problem related to our work is to design a polynomial time algorithm
of identity testing for the general low degree
$\Sigma\Pi\Sigma\Pi(k)$ circuits.

%\begin{bfseries}
%\Large Appendix
%\end{bfseries}

%\vspace{1cm}

%\begin{bfseries}
%\large
%\end{bfseries}

%\vspace{5mm}

\end{document}